\documentclass[submission,copyright,creativecommons]{eptcs}
\providecommand{\event}{ADG 2023} 

\usepackage{iftex}

\ifpdf
  \usepackage{underscore}         
  \usepackage[T1]{fontenc}        
\else
  \usepackage{breakurl}           
\fi
\usepackage{hyperref}
\makeatletter
\def\UrlAlphabet{%
\do\a\do\b\do\c\do\d\do\e\do\f\do\g\do\h\do\i\do\j%
\do\k\do\l\do\m\do\n\do\o\do\p\do\q\do\r\do\s\do\t%
\do\u\do\v\do\w\do\x\do\y\do\z\do\A\do\B\do\C\do\D%
\do\E\do\F\do\G\do\H\do\I\do\J\do\K\do\L\do\M\do\N%
\do\O\do\P\do\Q\do\R\do\S\do\T\do\U\do\V\do\W\do\X%
\do\Y\do\Z}
\def\UrlDigits{\do\1\do\2\do\3\do\4\do\5\do\6\do\7\do\8\do\9\do\0}
\g@addto@macro{\UrlBreaks}{\UrlOrds}
\g@addto@macro{\UrlBreaks}{\UrlAlphabet}
\g@addto@macro{\UrlBreaks}{\UrlDigits}
\makeatother

\newtheorem{theorem}{Theorem}

\newtheorem{definition}[theorem]{Definition}
\newtheorem{example}[theorem]{Example}

\newtheorem{problem}[theorem]{Problem}

\newtheorem{remark}[theorem]{Remark}

\newenvironment{proof}{\noindent{\em Proof:}}{$\Box$~\\}
\usepackage{amssymb}
\usepackage{mathrsfs}
\usepackage{amsfonts}
\usepackage{amsmath,amssymb}
\usepackage{graphicx,float}
\usepackage{enumerate}
\usepackage{multicol}
\usepackage{color}
\usepackage[ruled]{algorithm2e}
\usepackage{amsmath}
\usepackage{graphicx}%
\allowdisplaybreaks
\providecommand{\U}[1]{\protect\rule{.1in}{.1in}}

\newcommand{\mult}{{\texttt{mult}}}
\newenvironment{glists}[4]{
\begin{list}{}{
\setlength{\labelwidth}{#2}
\setlength{\labelsep}{#3}
\setlength{\leftmargin}{#1}
\addtolength{\leftmargin}{\labelwidth}
\addtolength{\leftmargin}{\labelsep}
\setlength{\parsep}{#4}
\setlength{\topsep}{\parsep}
\setlength{\itemsep}{\parsep}
\setlength{\listparindent}{0in}
}
}{
\end{list}
}
\newcommand{\iteml}[1]{\item[#1 \hfill]}

\title{Improving Angular Speed Uniformity\\by Piecewise Radical Reparameterization}
\author{Hoon Hong\institute{Department of Mathematics\\
North Carolina State University\\
Box 8205, Raleigh, NC 27695, USA}\email{hong@ncsu.edu}
\and Dongming Wang
\institute{LMIB -- IAI -- School of Mathematical Sciences\\
Beihang University\\
Beijing 100191, China}
\email{Dongming.Wang@cnrs.fr}
\and Jing Yang\thanks{Corresponding author.}\institute{SMS -- HCIC -- School of Mathematics and Physics\\
Center for Applied Mathematics of Guangxi\\
Guangxi Minzu University\\
Nanning 530006, China}\email{yangjing0930@gmail.com}
}
\def\titlerunning{Piecewise Radical Reparameterization}
\def\authorrunning{H. Hong, D. Wang \& J. Yang}
\begin{document}
\maketitle

\begin{abstract}
For a rational parameterization of a curve, it is desirable
that  its angular speed is as uniform as possible. Hence, given a  rational parameterization,  one wants to find {\em re}-parameterization with better uniformity.  One natural way is to use {\em piecewise} rational reparameterization.  However, it turns out that
the piecewise rational reparameterization does {\em not} help  when the angular speed of the given rational parameterization is zero at some points on the curve.
In this paper, we show how to overcome the challenge by  using piecewise {\em radical}
reparameterization.
\end{abstract}

\section{Introduction}

Parametric curves and surfaces are fundamental objects that are most
frequently used in computer aided geometric design. A given curve or surface
may have many different parameterizations, of which some may possess better
properties and thus are more suitable for certain applications than the
others. Thus, one often needs to convert one parameterization into
another, i.e., to {\em re}-parameterize the given parameterization (see, e.g.,
\cite{CFMS2001C,F2006R,F1997O,J1997A,K1991C,LZZL2005C,PB1989C,SV2001O,SWP2008R}%
). In this paper, we focus our investigation on an important class of
parameterizations, called uniform (angular-speed) parameterizations, where the
distribution of points are determined by the local curvature and show how to
construct such reparameterizations for a specific class of curves.

Uniform parameterization has been studied in a series of papers (see
\cite{HWY2013A,K1991C,PB1989C,YWH2012IC0,YWH2012IMUP,YWH2013I,YWH2013IC1} and
references therein). The authors have defined a function of angular speed
uniformity to measure the quality of any given parameterization of a plane
curve and proposed a method to compute its uniform reparameterization.
However, the computed reparameterization is irrational in most cases (with
straight lines as exceptions). For the sake of efficiency, a framework has
been proposed for the computation of rational approximations of uniform
parameterizations \cite{HWY2013A}. Four different methods of
reparameterization (i.e., optimal reparameterization with fixed degree,
$C^{0}$ and $C^{1}$ optimal piecewise reparameterization, and nearly optimal
$C^{1}$ piecewise reparameterization) have been integrated into this
framework. They have also been generalized to compute uniform quasi-speed
reparameterizations of parametric curves in $n$-dimensional space.

However, there is still a major challenge: all the above-mentioned methods  do {\em not} work well   when the angular speed of the given rational parameterization is zero at some points on the curve. This is due to an intrinsic
property of the angular speed function \cite{YWH2013I}:

\begin{quote}
\emph{Let $\omega_{p}$ be an angular speed function of a curve $p$ and $r$ be
a proper transformation. Then
\begin{equation}
\omega_{p\circ r}=(\omega_{p}\circ r)\cdot r^{\prime}.\label{prop:w}%
\end{equation}
}
\end{quote}

\noindent Uniformizing the angular speed can be seen as modifying the angular speed
value at each point under the constraint \eqref{prop:w} iteratively until all
the values are equal to the average. However, the constraint indicates that
$\omega_{p\circ r}$ will never reach the average value for any rational
$r^{\prime}$ when $\omega_{p}(t)=0$ for some $t$.

In this paper, we propose to  overcome  the challenge by  using  {\em radical}
transformations instead of rational ones.  We show that radical
transformations allow one    to increase the angular speed toward the
average value  at the points where the angular speed is zero.  Then   we adapt  the idea of piecewise M\"obius
transformation from \cite{YWH2012IC0} and  the  strategies  in \cite{YWH2013IC1}
to optimally improve the uniformity of angular
speed.

Experiments show that the proposed approach can improve the angular speed
uniformity significantly when  the angular speed of the given
parameterization vanishes at some point on the curve.

The rest of the paper is structured as follows. In Section \ref{sec:problem},
we formulate the problem precisely.   For this, we also introduce all the needed
notations and notions. In  Section \ref{sec:PRR}, we develop mathematical theory
to tackle the problem. In particular, we show how to use  piecewise radical  transformation to
transform an angular speed function with zeros into one without zero in such a way
that the parameters involved are also optimized. In Section
\ref{sec:AlgorithmExample}, we summarize the theoretical results into an algorithm
and illustrate its performance  on  an example. In   Section \ref{sec:implementation}, we briefly discuss implementational  issues/suggestions when floating point arithmetic
is used.

\section{Problem}

\label{sec:problem}

Consider a regular parametric curve
\[
p=(x_{1}(t),\ldots,x_{n}(t)):\,\mathbb{R}\mapsto\mathbb{R}^{n}.
\]
Its \emph{angular speed }$\omega_{p}$ is given by the following expression
(see \cite{HWY2013A}):\footnote[1]{The concept of angular speed is defined in the same manner as the one in physics but for $p'(t)$. The reason is to make the angular speed independent of the origin.}%
\begin{equation}
\omega_{p}=\dfrac{\sqrt{\sum\limits_{1\le i< j\le n}\left\vert
\begin{array}
[c]{cc}%
x_{i}^{\prime\prime} & x_{j}^{\prime\prime}\\
x_{i}^{\prime} & x_{j}^{\prime}%
\end{array}
\right\vert ^{2}}}{\sum\limits_{i=1}^nx_{i}^{\prime}{}^{2}}. \label{eq:w}%
\end{equation}
Recall that the mean $\mu_{p}$ and the variation $\sigma_{p}^{2}$ of
$\omega_{p}$ are given by
\[
\mu_{p}=\int_{0}^{1}{\omega_{p}(t)}\,dt,\quad\mbox{and}\quad\sigma_{p}%
^{2}=\int_{0}^{1}{(\omega_{p}(t)-\mu_{p})^{2}}\,dt.
\]

\begin{definition}
The \emph{angular speed uniformity} $u_{p}$ of a parameterization $p$ is
defined as
\begin{equation}
u_{p}=\left\{
\begin{array}
[c]{ll}%
\dfrac{1}{1+\sigma_{p}^{2}/\mu_{p}^{2}} &\quad \text{if }\mu_{p}\neq0,\\
1 &\quad \text{otherwise.}%
\end{array}
\right.  \label{eq:up}%
\end{equation}

\end{definition}

\begin{example}
[Running]\label{ex:ex1}
Consider the parametric curve  $p=(t,t^{3})$. Then
$$\omega_{p}=\dfrac{6\,t}{9\,t^{4}%
+1},$$
$\mu_p\doteq1.249$ and $u_{p}\doteq0.846$. The goal is to find a proper parameter transformation $r$ over $[0,1]$ in order to increase the uniformity.

\end{example}

Recall the following results from \cite{HWY2013A}.  For any proper
parameter transformation $r$ over $[0,1]$, we have
\begin{equation}\label{eqs:wpr}
\omega_{p\circ r}(s)=(\omega_{p}\circ r)(s)\cdot r^{\prime}(s)
\end{equation}
and
\begin{equation}
u_{p\circ r}=\dfrac{\mu_{p}^{2}}{\eta_{p,r}},\quad\mbox{where}\quad\eta
_{p,r}=\int_{0}^{1}\dfrac{\omega_{p}^{2}}{(r^{-1})^{\prime}}%
(t)\,dt.\label{eq:upr}%
\end{equation}

By \cite[Theorem 2]{YWH2013I}, one can construct a uniform
reparameterization from $p$, but such a reparameterization is irrational in
most cases. Therefore, we proposed several methods in \cite{HWY2013A} to
improve the angular speed uniformity by computing piecewise rational
reparameterizations. However, those methods are not applicable to curves whose
angular speed may vanish over $[0,1]$. Intuitively speaking, uniformizing the
angular speed over $[0,1]$ can be viewed as getting all the values of
$\omega_{p}(t)$ (for all $t\in[0,1]$) as close to $\mu_{p}$ as possible.

If $r$ is a continuous rational function over $[0,1]$, then $r^{\prime}$ is
bounded. Suppose that $\omega
_{p}(t_{0})=0$ for some $t_{0}\in[0,1]$ and $\mu_{p}\neq0$. Then by \eqref{eqs:wpr}, there must
exist some $s_{0}\in[0,1]$ such that $\omega_{p\circ r}(s_{0})=0$, which is
not close to $\mu_{p}$ at all. This makes rational proper parameter
transformations invalid. In what follows, we resort to radical transformations
and develop a new approach to uniformize the angular speed of parametric
curves which has zeros over $[0,1]$.

Let $p$ be a parametric curve. Without loss of generality, we assume that
\[
p^{\prime}(t)=(x_{1}^{\prime}(t),\ldots,x_{n}^{\prime}(t))=\Big(\frac
{X_{1}(t)}{W(t)},\cdots,\frac{X_{n}(t)}{W(t)}\Big),
\]
where~$X_{i}(t),W(t)\in\mathbb{R}[t]$~and~$\gcd(X_{1}(t),\ldots,X_{n}%
(t),W(t))=1$. One can verify that
$$
\omega_{p}=\dfrac{\sqrt{F}}{\sum\limits_{i}X_{i}{}^{2}}, 
\quad\text{where}\quad
F=\sum\limits_{i\neq j}\left\vert
\begin{array}
[c]{cc}%
X_{i}^{\prime} & X_{j}^{\prime}\\
X_{i} & X_{j}%
\end{array}
\right\vert ^{2}. 
$$
Let $F$ be written as $F=\left(\prod_{i=0}^{k}(t-\tilde{t}_{i}%
)^{2\,\mu_{i}}\right)\zeta(t)$ for positive $k$.
Note that $\{\tilde{t}_{i}:\,\tilde{t}_{i}<\tilde{t}_{i+1}\ \text{for}\ 0\le i<k\}$ contains all the zeros of $F$
over $[0,1]$. It is allowed that some $t_i$'s are not the roots of $F$.
The positive integer $\mu_{i}\in\mathbb{N}$ is called the \emph{multiplicity} of
$\tilde{t}_{i}$ in $\omega_{p}$ and denoted by ${\mathtt{mult}}(\omega
_{p},\tilde{t}_{i})$. If $\omega_{p}(\tilde{t}_{i})\neq0$, then
${\mathtt{mult}}(\omega_{p},\tilde{t}_{i})=0$.

Let
\[
T=(t_{0},\ldots, t_{N}), \quad S=(s_{0},\ldots, s_{N}),\quad Z=(z_{0},\ldots,
z_{N}),\quad\alpha=(\alpha_{0},\ldots,\alpha_{N-1})
\]
be sequences such that

\begin{itemize}
\item $0=t_{0}<\cdots<t_{N}=1$, $0=z_{0}<\cdots<z_{N}=1$, $0=s_{0}%
<\cdots<s_{N}=1$, $0< \alpha_{i}< 1$;

\item at most one of $\omega_{p}(t_{i})=0$ and $\omega_{p}(t_{i+1})=0$ holds
for $0\le i<N$, that is, the successive appearance of two zeros of $\omega
_{p}$ are not allowed;

\item the multiplicity of $t_{i}$ in $\omega_{p}$ is $\mu_{i}$;

\item $\omega_{p}(t)\neq0$ for all $t\in(t_{i}, t_{i+1})$.
\end{itemize}


\begin{definition}
[Elementary Piecewise Radical Transformation]\label{def:prt} Let $p$ be a
parametric curve with $T, S$ defined above. Then $\varphi$ is called an
\emph{elementary piecewise radical transformation} associated to $p$ if
$\varphi$ has the following form:
\vspace{-.5em}
\[
\varphi(s) = \left\{
\begin{array}
[c]{cll}%
\vdots~ &  & \\
\varphi_{i}(s) & \quad\text{if} & s \in[s_{i} ,s_{i+1}],\\
\vdots~ &  &
\end{array}
\right.\vspace{-.5em}
\]
where
\begin{equation}
\label{eq:varphii}\varphi_{i}(s)= \left\{
\begin{array}
[c]{ll}%
t_{i}+\Delta t_{i}\sqrt[\mu_{i}+1]{ \tilde{s} }\quad & \mbox{if~}\omega
_{p}(t_{i})=0;\\[8pt]%
t_{i}+\Delta t_{i}(1-\sqrt[\mu_{i+1}+1]{1- \tilde{s} })\quad &
\mbox{if~}\omega_{p}(t_{i+1})=0;\\[8pt]%
t_{i}+\Delta t_{i}\cdot\tilde{s}\quad & \mbox{otherwise,}
\end{array}
\right.
\end{equation}
and $\Delta t_{i}=t_{i+1}-t_{i},\, \Delta s_{i}=s_{i+1}-s_{i}$, $\tilde
{s}=(s-s_{i})/\Delta s_{i}$.
\end{definition}

\begin{remark}\
\begin{enumerate}
\item It can be verified that $\varphi(s_{i})=t_{i}$ and $\varphi
(s_{i+1})=t_{i+1}$, which implies that $\varphi$ is with $C^{0}$ continuity.

\item It is allowed that more than one intermediate point lie between two
zeros of $\omega_{p}$ because it can reduce the number of radical pieces and
thus enhance the efficiency of generating points with the new parameterization.
\end{enumerate}
\end{remark}

It can be shown that $\omega_{p\circ\varphi}(s)\neq0$ (see Theorem
\ref{th:avnonzero}). Next let $q=p\circ\varphi$ and thus $q$ has no inflation
point. We adapt the reparameterization methods from \cite{YWH2013IC1} to
increase the uniformity of $\omega_{q}$ to any value close to 1. For this
purpose, we recall the following piecewise M\"{o}bius transformation.

\begin{definition}
[Piecewise M\"{o}bius Transformation]\label{def:pmt} Let $p$ be a parametric
curve with $S, Z, \alpha$ defined above. Then $m$ is called a \emph{piecewise
M\"{o}bius transformation} associated to $p$ if $m$ has the following form:
\[
m(z) = \left\{
\begin{array}
[c]{cll}%
\vdots~ &  & \\
m_{i}(z) & \quad\text{if} & z \in[z_{i} ,z_{i+1}],\\
\vdots~ &  &
\end{array}
\right.
\]
where
\begin{equation}
\label{eqs:moebiuspiece}m_{i}(z)=s_{i}+\Delta s_{i}\cdot\dfrac{(1-\alpha
_{i})\tilde{z}}{(1-\alpha_{i})\tilde{z}+\alpha_{i}(1-\tilde{z})}%
\end{equation}
and $\Delta z_{i}=z_{i+1}-z_{i},\, \Delta s_{i}=s_{i+1}-s_{i}$, $\tilde
{z}=(z-z_{i})/\Delta z_{i}$.
\end{definition}

The problem addressed in this paper may be formulated as follows.

\begin{problem}
Given a parametric curve $p$ with $\omega_{p}(t)=0$ for some $t\in[0,1]$, find
a radical piecewise transformation $\varphi$ and an optimal piecewise M\"{o}bius
transformation $m$ over $[0,1]$ such that
\begin{itemize}
\item $u_{p\circ\varphi\circ{m}}\doteq1$;\medskip
\item $\forall s\in[0,1]$, $\omega_{p\circ\varphi\circ{m}}(s)\neq0$.
\end{itemize}
\end{problem}


\section{Theory}

\label{sec:PRR}

\subsection{Property of $\varphi$}

\begin{theorem}
\label{th:avnonzero} For any $s\in[0,1]$, $\omega_{p\circ\varphi}(s)\neq0$.
\end{theorem}
\begin{proof}
Taking derivative of $\varphi_i$, we have
\begin{equation*}
\varphi_i'(s)=\left\{
\begin{array}{ll}
\dfrac{\Delta t_i}{\mu_i+1}
        \cdot \dfrac{1}{\sqrt[\mu_i+1]{ \tilde{s}^{\mu_i}}}
        \cdot \dfrac{1}{\Delta s_i}\quad
&\mbox{if~}\omega_p(t_i)=0;\\
\dfrac{\Delta t_i}{\mu_{i+1}+1}
        \cdot\dfrac{1}{\sqrt[\mu_{i+1}+1]{(1- \tilde{s} )^{\mu_{i+1}}}}
        \cdot\dfrac{1}{\Delta s_i}\quad
&\mbox{if~}\omega_p(t_{i+1})=0;\\
\dfrac{\Delta t_i}{\Delta s_i}\quad
&\mbox{otherwise}.
\end{array}
\right.
\end{equation*}

Next we show that in the above three cases, $\omega_{p\circ \varphi}(s)\ne0$.

\begin{enumerate}[{Case} 1:]
\item
$\omega_p(t_i)=0.$

Assume that $\mu_i=\mult(\omega_p, t_i)$. Then $\omega_p$
can be written as
\[
\omega_p=|t-t_i|^{\mu_i}\cdot\tilde{\zeta}(t),\] where
$\tilde{\zeta}(t)>0$ for $t\in [t_i, t_{i+1}]$. Therefore,
\begin{align*}
\omega_{p\circ \varphi}(s)&=~|\varphi_i(s)-t_i|^{\mu_i}
        \cdot(\tilde{\zeta}\circ\varphi_i)(s)
        \cdot\varphi_i'(s)\\
&=(\Delta t_i \tilde{s}^{\frac{1}{\mu_i+1}})^{\mu_i}
        \cdot(\tilde{\zeta}\circ\varphi_i)(s)
        \cdot \left[\dfrac{\Delta t_i}{\mu_i+1}
        \cdot \dfrac{1}{\sqrt[\mu_i+1]{ \tilde{s}^{\mu_i}}}
        \cdot \dfrac{1}{\Delta s_i}\right]\\
&=\dfrac{\Delta t_i^{\mu_i+1}}{\mu_i+1}
        \cdot(\tilde{\zeta}\circ\varphi_i)(s)
        \cdot  \dfrac{1}{\Delta s_i}\\
&=\dfrac{\Delta t_i^{\mu_i+1}}{\mu_i+1}
        \cdot  \dfrac{1}{\Delta s_i}
        \cdot\tilde{\zeta}(t)
\neq0
\end{align*}
for $s\in[s_i,s_{i+1}]$.

\item
$\omega_p(t_{i+1})=0$.

Assume that $\mu_{i+1}=\mult(\omega_p, t_{i+1})$.
Then $\omega_p$ can be written as
\[
\omega_p=|t_{i+1}-t|^{\mu_{i+1}}\cdot\tilde{\zeta}(t),\] where
$\tilde{\zeta}(t)>0$ for $t\in [t_i, t_{i+1}]$. Therefore,
\begin{align*}
\omega_{p\circ \varphi}(s)&=\,|t_{i+1}-\varphi_i(s)|^{\mu_{i+1}}
        \cdot(\tilde{\zeta}\circ\varphi_i)(s)
        \cdot\varphi_i'(s)\\
&=\,\dfrac{\Delta t_i}{\mu_{i+1}+1}
        \cdot\dfrac{1}{\sqrt[\mu_{i+1}+1]{(1-\tilde{s})^{\mu_{i+1}}}}
        \cdot \dfrac{1}{\Delta s_i}\\
&=\,[\Delta t_i (1-\tilde{s})^{\frac{1}{\mu_{i+1}+1}}]^{\mu_{i+1}}\\
&\ \ \ \         \cdot(\tilde{\zeta}\circ\varphi_i)(s)
        \cdot \dfrac{\Delta t_i}{\mu_{i+1}+1}
        \cdot\dfrac{1}{\sqrt[\mu_{i+1}+1]{(1-\tilde{s})^{\mu_{i+1}}}}
        \cdot \dfrac{1}{\Delta s_i}\\
&=\,\dfrac{\Delta t_i^{\mu_{i+1}+1}}{\mu_{i+1}+1}
        \cdot(\tilde{\zeta}\circ\varphi_i)(s)
        \cdot \dfrac{1}{\Delta s_i}\\
&=\,\dfrac{\Delta t_i^{\mu_{i+1}+1}}{\mu_{i+1}+1}
        \cdot\tilde{\zeta}(t)
        \cdot \dfrac{1}{\Delta s_i}\neq0 .
\end{align*}

\item $\omega_p(t_i)\omega(t_{i+1})\neq0$.

Combining $\omega_p(t)\ne0$ for $t\in[t_i,t_{i+1}]$, ${\Delta t_i>0}$ and ${\Delta s_i>0}$, we have
\[\omega_{p\circ \varphi}(s)=(\omega_{p}\circ \varphi)(s)\cdot \varphi'(s)=\omega_p(t)\cdot\dfrac{\Delta t_i}{\Delta s_i}\ne0. \]
\end{enumerate}

To sum up, we have $\omega_{p\circ \varphi}(s)\neq 0$ when $s\in[s_i,s_{i+1}]$.
\end{proof}

\begin{example}[Continued from Example \ref{ex:ex1}]
\label{ex:ex2} For the cubic curve $p=(t, t^{3})$ whose angular speed
is
$
\omega_{p}=\dfrac{6\,t}{9\,t^{4}+1},
$
it is easy to see that $t=0$ is a zero of $\omega_{p}$ with multiplicity $1$. Let
$T=(0,1)$ and $S=(0,1)$. Then the constructed $\varphi$ is $\varphi
(s)=\sqrt{s}$. It follows that
\[
\omega_{p\circ\varphi}(s)=(\omega_{p}\circ\varphi)(s)\cdot\varphi^{\prime}(s)
=\dfrac{6\,\sqrt{s}}{9\,s^{2}+1}\cdot\dfrac{1}{2\,\sqrt{s}}=\dfrac{3}%
{9\,s^{2}+1}%
\]
which is nonzero over $[0,1]$.
\end{example}

\begin{remark}
It may be further deduced that $\omega_{p\circ\varphi}(s)$ is discontinuous at
$s=s_{i}$.
\end{remark}

\subsection{Choice of $T$}

By Definition \ref{def:prt}, $T$ should contain all the zeros of $\omega_{p}$
over $[0,1]$ and some intermediate points in the subintervals separated by the
zeros of $\omega_{p}$. One question is how to choose intermediate points to
make the uniformity improvement as significant as possible. In this
subsection, we present a strategy similar to the one introduced in
\cite{YWH2013IC1} for determining such points.

Recall \cite[Theorem 2]{YWH2013I} which states that the uniformizing parameter transformation $r_{p}$ of $p$
satisfies
\[
(r_{p})^{-1}=\int_{0}^{t}\,\omega_{p}(\gamma)d\gamma/\mu_{p}.
\]
Let $\varphi$ be a piecewise radical transformation associated to $p$. If
$(r_{p})^{-1}$ and $\varphi^{-1}$ share some common properties, we say
informally that $r_{p}$ and $\varphi$ are similar to each other.

First of all, the following can be derived:
\begin{equation}
\label{eqs:inversephi}\varphi^{-1}(t)= \left\{
\begin{array}
[c]{ll}%
s_{i}+\Delta s_{i}\cdot\tilde{t}^{\mu_{i}+1}\quad~ & \mbox{if~}\omega
_{p}(t_{i})=0;\\[8pt]%
s_{i}+\Delta s_{i}\cdot[1-(1-\tilde{t})^{\mu_{i+1}+1}]\quad~ &
\mbox{if~}\omega_{p}(t_{i+1})=0;\\[8pt]%
s_{i}+\Delta s_{i}\cdot\tilde{t}\qquad\quad~ & \mbox{otherwise},
\end{array}
\right.
\end{equation}
where $\tilde{t}=(t-t_{i})/\Delta t_{i}$. Furthermore,
\begin{align*}
\left[  \varphi^{-1}\right]  ^{\prime}  &  =~ \left\{
\begin{array}
[c]{ll}%
\dfrac{\Delta s_{i}}{\Delta t_{i}}\cdot(\mu_{i}+1)\cdot\tilde{t}^{\mu_{i}}
~\quad\qquad & \mbox{if~}\omega_{p}(t_{i})=0;\\[8pt]%
\dfrac{\Delta s_{i}}{\Delta t_{i}}\cdot(\mu_{i+1}+1)\cdot(1-\tilde{t}%
)^{\mu_{i+1}}\qquad\quad~ & \mbox{if~}\omega_{p}(t_{i+1})=0;\\[8pt]%
\dfrac{\Delta s_{i}}{\Delta t_{i}}\qquad\quad~ & \mbox{otherwise};
\end{array}
\right. 
\\[5pt]
\left[  \varphi^{-1}\right]  ^{\prime\prime}  &  =~ \left\{
\begin{array}
[c]{ll}%
\dfrac{\Delta s_{i}}{\Delta t_{i}^{2}}\cdot(\mu_{i}+1)\mu_{i}\cdot\tilde
{t}^{\mu_{i}-1}\quad & \mbox{if~}\omega_{p}(t_{i})=0;\\[8pt]%
-\dfrac{\Delta s_{i}}{\Delta t_{i}^{2}}\cdot(\mu_{i+1}+1)\mu_{i+1}%
\cdot(1-\tilde{t})^{\mu_{i+1}-1}\quad & \mbox{if~}\omega_{p}(t_{i+1}%
)=0;\\[8pt]%
0\qquad\quad~ & \mbox{otherwise}.
\end{array}
\right.  
\end{align*}

Note that $\varphi$ has the properties listed below.

\begin{itemize}
\item $\varphi^{-1}(0)=0$, $\varphi^{-1}(1)=1$.\smallskip

\item $\varphi_{i}^{-1}$ is monotonic over $(t_{i}, t_{i+1})$ because
$(\varphi_{i}^{-1})^{\prime}(t)\geq0$ for all $t\in(t_{i}, t_{i+1})$; since
$\varphi^{-1}$ is continuous over $[0,1]$, $\varphi^{-1}$ is monotonic over
$[0,1]$.\smallskip

\item $[\varphi_{i}^{-1}]^{\prime}$ is monotonic over $(t_{i}, t_{i+1})$
because $[\varphi_{i}^{-1}]^{\prime\prime}$ has a constant sign over $(t_{i},
t_{i+1})$.
\end{itemize}

The above properties indicate that $\varphi$ is composed of some monotonically
increasing convex or concave pieces.
Moreover, it can be verified that

\begin{itemize}
\item $r_{p}^{-1}(0)=\int_{0}^{0}\omega_{p}(\gamma)\,d\gamma/\mu_{p}=0$,
$r_{p}^{-1}(1)=\int_{0}^{1}\omega_{p}(\gamma)\,d\gamma/\mu_{p}=1$;\smallskip

\item $r_{p}^{-1}$ is monotonic over $[0,1]$ because $(r_{p}^{-1})^{\prime
}(t)=\omega_{p}(t)/\mu_{p}\geq0$.
\end{itemize}

One may observe that $\varphi$ shares the first two properties with $r_{p}$.
If $r_{p}$ possesses the third property of $\varphi$, then $r_{p}$ and
$\varphi$ are expected to be similar. This inspires us to divide $[0,1]$ into
some monotonic intervals of $(r_{p}^{-1})^{\prime}(t)$ (i.e., $\omega_{p}$).
Thus we may try to choose the intermediate $t_{i}$ in $T$ by solving
$$
\omega_{p}(t_{i})\omega_{p}^{\prime}(t_{i})=0.
$$
Note that $\omega_{p}(t)$ is nonnegative. Thus $0$ is the local minimum value
of $\omega_{p}$. In this sense, $T$ consists of all the local extreme points
of $\omega_{p}$ and the two boundary points of the unit interval.

With the above operation, $r_{p}$ is divided into some monotonically
increasing/decreasing convex or concave pieces with each piece having a
corresponding one in $\varphi$. Therefore, $T$ can be obtained by collecting
and inserting the zeros of $\omega_{p}$ and $\omega_{p}^{\prime}$
into $[0,1]$ in order.

\begin{example}[Continued from Example \ref{ex:ex2}]\label{ex:ex3}
One may compute that
\[
\omega_{p}'(t)=-\dfrac{6\,(27\,t^4-1)}{(9\,t^4+1)^2}.
\]
Then the solution of  $\omega_{p}(t)\omega_p'(t)=0$
over $[0,1]$ gives us a partition of $[0,1]$, i.e.,
\[T\doteq(0,0.439,1).\]
Furthermore, one may check that the multiplicities of $t_0,t_1,t_2$ as roots of $\omega_p$ are $1,0$ and $0$, respectively.
\end{example}

\subsection{Determination of $S$}

Once a partition $T$ of $[0,1]$ is obtained, one can compute the sequence $S$
in various ways. In this subsection, we present an optimization strategy for
the computation of $S$.

When $T$ is fixed, $u_{p\circ\varphi}$ becomes a function of $s_{i}%
~(i=1,\ldots,N-1)$. The following theorem provides a formula for computing the
optimal values for $s_i$'s.

\begin{theorem}
\label{th:optimaluniformity} The uniformity $u_{p\circ\varphi}$ reaches the
maximum when
\begin{equation}
\label{eqs:si}s_{i}=s_{i}^{*}=\dfrac{\sum_{k=0}^{i-1}\sqrt{L_{k}}}{\sum
_{k=0}^{N-1}\sqrt{L_{k}}},
\end{equation}
where{
\begin{equation}
\label{eqs:mk}L_{k}=\left\{
\begin{array}
[c]{ll}%
{\Delta t_{k}}\int_{t_{k}}^{t_{k+1}}\dfrac{\omega_{p}^{2}(t)}{(\mu
_{k}+1)\tilde{t}^{\mu_{k}}}\,dt\quad & \mbox{if}\quad\omega_{p}(t_{k})=0;\\[10pt]
{\Delta t_{k}}\int_{t_{k}}^{t_{k+1}}\dfrac{\omega_{p}^{2}(t)}{(\mu
_{k+1}+1)(1-\tilde{t})^{\mu_{k+1}}}\,dt\quad & \mbox{if}\quad\omega
_{p}(t_{k+1})=0;\\[10pt]
\Delta t_{k}\int_{t_{k}}^{t_{k+1}}\omega_{p}^{2}(t)\,dt\quad &
\mbox{otherwise.}
\end{array}
\right.
\end{equation}
}
The maximum value of $u_{p\circ\varphi}$ is
\[
u_{p\circ\varphi}^{*}=\mu_{p}^{2}\big/\eta_{p,\varphi}^{*},\quad
\mbox{where}\quad\eta_{p,\varphi}^{*}=\left(  \sum_{i=0}^{N-1}\sqrt{L_{i}%
}\right)  ^{2}.
\]
\end{theorem}

\begin{proof}
Recall \eqref{eq:upr}. Since $\mu_p$ is a constant for any given
$p$, the problem of maximizing $u_{p\circ \varphi}$ can be reduced
to that of minimizing
\[\eta_{p,\varphi}=\int_0^1\dfrac{\omega_p^2}{(\varphi^{-1})'}(t)\,dt
=\sum_{i=0}^{N-1}\int_{t_i}^{t_{i+1}}\dfrac{\omega_p^2}{(\varphi_i^{-1})'}(t)\,dt.\]
We first simplify each component in the above equation. Denote
$\int_{t_i}^{t_{i+1}}\dfrac{\omega_p^2}{(\varphi_i^{-1})'}(t)\,dt$
by $I_i$. Note that
\[
\dfrac{1}{(\varphi_i^{-1})'(t)}=\left\{
\begin{array}{ll}
\dfrac{\Delta t_i}{\Delta s_i}\cdot \dfrac{1}{\mu_i+1}\cdot\dfrac{1}{\tilde{t}^{\mu_i}}\quad&\mbox{if}\quad \omega_p(t_i)=0;\\
\dfrac{\Delta t_i}{\Delta s_i}\cdot \dfrac{1}{\mu_{i+1}+1}\cdot\dfrac{1}{(1-\tilde{t})^{\mu_{i+1}}}\quad&\mbox{if}\quad \omega_p(t_{i+1})=0;\\
\dfrac{\Delta t_i}{\Delta s_i}\quad& \mbox{otherwise}.
\end{array}
\right.
\]
When $\omega_p(t_i)=0$,
\[I_i=\dfrac{\Delta t_i}{\Delta s_i}\cdot\dfrac{1}{\mu_i+1}\int_{t_i}^{t_{i+1}}\dfrac{\omega_p^2}{\tilde{t}^{\mu_i}}\,dt=L_i/\Delta s_i.\]
Similarly, when $\omega_p(t_{i+1})=0$,
\[I_i=\dfrac{\Delta t_i}{\Delta s_i}\cdot\dfrac{1}{\mu_{i+1}+1}\int_{t_i}^{t_{i+1}}\dfrac{\omega_p^2}{(1-\tilde{t})^{\mu_{i+1}}}\,dt=L_i/\Delta s_i.\]
When $\omega_p(t_i)\cdot\omega_p(t_{i+1})\neq0$,
\begin{equation*}
I_i=\dfrac{\Delta t_i}{\Delta s_i}\cdot\int_{ _i}^{t_{i+1}}\omega_p^2\,dt={L_i}/{\Delta s_i}.
\end{equation*}
It is obvious that $\eta_{p,\varphi}=\sum_{i=0}^{N-1}I_i>0$; it
increases to $+\infty$ when $s_i$ approaches the
boundary of the feasible set of parameters. Now we compute the
extrema of $\eta_{p,\varphi}$. Let
\[\dfrac{\partial \eta_{p,\varphi}}{\partial s_i}=0,\]
i.e.,
\[\dfrac{L_i}{\Delta s_i^2}-\dfrac{L_{i-1}}{\Delta s_{i-1}^2}=0,\]
where $L_i$ is as in \eqref{eqs:mk}. Solving the above equation, we
obtain
\[\Delta s_i=\Delta s_i^*=\Delta s_0^*\sqrt{L_i/L_0}.\]
Note that $\sum_{i=0}^{N-1}\Delta s_i^*=1$. Thus
\[\Delta s_0^*=\left(\sum_{k=0}^{N-1}\sqrt{L_k/L_0}\right)^{-1},\quad s_i^*=\sum_{k=0}^{i-1}{\Delta s_k^*}=\dfrac{\sum_{k=0}^{i-1}\sqrt{L_k/L_0}}{\sum_{k=0}^{N-1}\sqrt{L_k/L_0}}=\dfrac{\sum_{k=0}^{i-1}\sqrt{L_k}}{\sum_{k=0}^{N-1}\sqrt{L_k}}.\]
Therefore,
\[
\Delta s_i^*=\dfrac{\sqrt{L_i}}{\sum_{k=0}^{N-1}\sqrt{L_k}}.
\]
Moreover, the optimal value of $\eta_{p,\varphi}$ is
\[\eta_{p,\varphi}=\eta_{p,\varphi}^*=\sum_{i=0}^{N-1}\dfrac{L_i}{\Delta s_i^*}=\sum_{k=0}^{N-1}\dfrac{L_i}{\dfrac{\sqrt{L_i}}{\sum_{i=0}^{N-1}\sqrt{L_k}}}
=\left(\sum_{k=0}^{N-1}\sqrt{L_k}\right)^2,\] from which it follows
that the optimal value of $u_{p\circ \varphi}$ is
\[u_{p\circ \varphi}=u_{p\circ \varphi}^*=\dfrac{\mu_p^2}{\eta_{p,\varphi}^*}=\dfrac{\mu_p^2}{\left(\sum_{k=0}^{N-1}\sqrt{L_k}\right)^2}.\]
The proof is completed.
\end{proof}

\begin{example}[Continued from Example \ref{ex:ex3}]\label{ex:ex4}
By using \eqref{eqs:mk}, we compute the values of $L_0$ and $L_1$ and obtain
\begin{align*}
L_0&\doteq0.439\int_{0}^{0.439}\dfrac{\left(\dfrac{6\,t}{9\,t^{4}+1}\right)^{2}}{2\cdot \dfrac{t-0}{0.439}}\,dt\doteq0.276,\\
L_1&\doteq(1-0.439)\int_{0.439}^{1}\left(\frac{6\,t}{9\,t^{4}+1}\right)^{2}\,dt\doteq0.590.
\end{align*}
By \eqref{eqs:si},
we have $s_1=0.406$. Thus $S\doteq(0,0.406,1)$. Furthermore, one may calculate the optimal value of $u_{p\circ\varphi}$ and obtain $u_{p\circ\varphi}^*\doteq0.932$.
\end{example}

\subsection{Determination of $Z$ and $\alpha$}

Once a partition $S$ of $[0,1]$ is obtained, one can compute the sequence $Z$.
In this subsection, we give explicit formulae for the optimal values of $Z$ and
$\alpha$ which are directly computed from the sequence $T$. For this purpose,
we first recall the the following result from \cite{YWH2012IC0}.

\begin{theorem}
\label{thm:optimalZalpha} Let $q$ be a rational parameterization such that
$\omega_{q}(s)\ne0$ over $[0,1]$ and $m$ be a piecewise M\"{o}bius
transformation determined by $S$, $Z$ and $\alpha$. For a given sequence $S$,
the uniformity $u_{q\circ m}$ reaches the maximum when
\begin{equation}
\label{eqs:alphaZ}\alpha_{i}=\alpha_{i}^{*}=\dfrac{1}{1+\sqrt{C_{i}/A_{i}}%
},\quad z_{i}=z_{i}^{*}=\dfrac{\sum_{k=0}^{i-1}\sqrt{M_{k}}}{\sum_{k=0}%
^{N-1}\sqrt{M_{k}}},%
\end{equation}
where{
\begin{align*}
A_{i}  &  =\int_{s_{i}}^{s_{i+1}}\omega_{q}^{2}\cdot(1-\tilde{s})^{2}ds, &
B_{i}  &  =\int_{s_{i}}^{s_{i+1}}\omega_{q}^{2}\cdot2\tilde{s}(1-\tilde
{s})ds,\\
C_{i}  &  =\int_{s_{i}}^{s_{i+1}}\omega_{q}^{2}\cdot\tilde{s}^{2}ds, & M_{k}
&  =\Delta s_{k}\left(  2\sqrt{A_{k}C_{k}}+B_{k}\right)  .
\end{align*}
Let $m^{*}$ be the piecewise M\"{o}bius transformation determined by $S$,
$Z^{*}$ and $\alpha^{*}$. Then the maximum value of $u_{q\circ m}$ is
$u_{q\circ m^{*}}=\mu_{q}^{2}/\eta_{q,m^{*}}$ where
\[
\eta_{q,m^{*}}=\left(  \sum_{i=0}^{N-1}\sqrt{M_{k}}\right)  ^{2}.%
\]}
\end{theorem}

\begin{remark}
Let $\varphi$ be an elementary radical transformation as in Definition
\ref{def:prt} and $q=p\circ\varphi$. Note that
\[
\mu_{q}=\int_{0}^{1}\omega_{q}ds=\int_{0}^{1}\omega_{p\circ\varphi}ds=\int%
_{0}^{1}(\omega_{p}\circ\varphi)(s)\cdot\varphi^{\prime}(s)ds=\int_{0}%
^{1}\omega_{p}dt=\mu_{p}.
\]
Thus
\begin{equation*}
u_{p\circ\varphi\circ m^{*}}=u_{q\circ m^{*}}=\mu_{q}%
^{2}/\eta_{q,m^{*}}=\mu_{p}^{2}\bigg/\left(  \sum_{i=0}^{N-1}\sqrt{M_{k}%
}\right)  ^{2}.
\end{equation*}

\end{remark}

Let $\varphi$ and $q$ be defined as before. By Theorem \ref{th:avnonzero},
$\omega_{q}\ne0$ over $[0,1]$. One may compute the optimal values of $S$, $\alpha$
and $Z$ by Theorems \ref{th:optimaluniformity} and \ref{thm:optimalZalpha}. However, $p\circ\varphi$ is a
composition of radical function and rational function and the composition will
cause an increase of complexity because $\omega_{q}$ is radical. In what
follows, we simplify the formulae for $A_{i},B_{i}$ and $C_{i}$ with the goal
of computing the values of $A_{i},B_{i}$ and $C_{i}$ directly from $p$.

The formula of $A_i~(0\le i\le N-1)$ is derived via the following steps:
\begin{align*}
A_{i}  &=\int_{s_i}^{s_{i+1}}\omega_q^2\cdot(1-\tilde{s})^2ds\\
&=\int_{s_i}^{s_{i+1}}[(\omega_p\circ\varphi)(s)]^2\cdot[\varphi'(s)]^2\cdot(1-\tilde{s})^2ds\\
&=\int_{s_i}^{s_{i+1}}[(\omega_p\circ\varphi)(s)]^2\cdot[\varphi'(s)]\cdot(1-\tilde{s})^2\left[\varphi'(s) ds\right]\\
&=\int_{t_i}^{t_{i+1}}\dfrac{\omega_p^2}{\left(\varphi^{-1}\right)'}(t)\cdot\left(1-\dfrac{\varphi^{-1}-s_i}{\Delta s_i}\right)^2dt\qquad\qquad\qquad\qquad\quad\text{by~\eqref{eqs:inversephi}}\\
&=\left\{
\begin{array}{ll}
\int_{t_i}^{t_{i+1}}\dfrac{\omega_p^2}{\left(\varphi^{-1}\right)'}(t)\cdot(1-\tilde{t}^{\mu_i+1})^2dt
&\text{if~}\omega(t_i)=0;\\
\int_{t_i}^{t_{i+1}}\dfrac{\omega_p^2}{\left(\varphi^{-1}\right)'}(t)\cdot(1-\tilde{t})^{2(\mu_{i+1}+1)}dt
&\text{if~}\omega(t_{i+1})=0;\\
\int_{t_i}^{t_{i+1}}\dfrac{\omega_p^2}{\left(\varphi^{-1}\right)'}(t)\cdot(1-\tilde{t})^2dt
&\text{otherwise};
\end{array}
\right.\\
&  =\left\{
\begin{array}
[c]{lll}%
\dfrac{\Delta t_{i}}{\Delta s_{i}}\int_{t_{i}}^{t_{i+1}}\dfrac{\omega_{p}%
^{2}(t)}{(\mu_{i}+1)\tilde{t}^{\mu_{i}}}\cdot(1-\tilde{t}^{\mu_{i}+1}%
)^{2}dt &  &\hspace{7em} \text{if~}\omega(t_{i})=0;\\[10pt]%
\dfrac{\Delta t_{i}}{\Delta s_{i}}\int_{t_{i}}^{t_{i+1}}\dfrac{\omega_{p}%
^{2}(t)}{\mu_{i+1}+1}\cdot(1-\tilde{t})^{\mu_{i+1}+2}dt &  &\hspace{7em} \text{if~}%
\omega(t_{i+1})=0;\\[10pt]%
\dfrac{\Delta t_{i}}{\Delta s_{i}}\int_{t_{i}}^{t_{i+1}}{\omega_{p}^{2}%
}(t)\cdot(1-\tilde{t})^{2}dt &  &\hspace{7em} \text{otherwise.}%
\end{array}
\right.
\end{align*}
Similarly, we have
\begin{align*}
B_{i}  &  =\left\{
\begin{array}
[c]{ll}%
\dfrac{\Delta t_{i}}{\Delta s_{i}}\int_{t_{i}}^{t_{i+1}}\dfrac{\omega_{p}%
^{2}(t)}{\mu_{i}+1}\cdot2\tilde{t}(1-\tilde{t}^{\mu_{i}+1})dt &
\hspace{2.3em}\text{if~}\omega(t_{i})=0;\\[10pt]%
\dfrac{\Delta t_{i}}{\Delta s_{i}}\int_{t_{i}}^{t_{i+1}}\dfrac{\omega_{p}%
^{2}(t)}{\mu_{i+1}+1}\cdot2[1-(1-\tilde{t})^{\mu_{i+1}+1}](1-\tilde{t}%
)^{}dt & \hspace{2.3em}\text{if~}\omega(t_{i+1})=0;\\[10pt]%
\dfrac{\Delta t_{i}}{\Delta s_{i}}\int_{t_{i}}^{t_{i+1}}{\omega_{p}^{2}%
}(t)\cdot(1-\tilde{t})^{2}dt & \hspace{2.3em}\text{otherwise;}%
\end{array}
\right.\\
C_{i}  &  =\left\{
\begin{array}
[c]{ll}%
\dfrac{\Delta t_{i}}{\Delta s_{i}}\int_{t_{i}}^{t_{i+1}}\dfrac{\omega_{p}%
^{2}(t)}{\mu_{i}+1}\cdot\tilde{t}^{\mu_{i}+2}dt & \text{if~}\omega
(t_{i})=0;\\[10pt]%
\dfrac{\Delta t_{i}}{\Delta s_{i}}\int_{t_{i}}^{t_{i+1}}\dfrac{\omega_{p}%
^{2}(t)}{(\mu_{i+1}+1)(1-\tilde{t})^{\mu_{i+1}}}\cdot[1-(1-\tilde{t}%
)^{\mu_{i+1}+1}]^{2}dt & \text{if~}\omega(t_{i+1})=0;\\[10pt]%
\dfrac{\Delta t_{i}}{\Delta s_{i}}\int_{t_{i}}^{t_{i+1}}{\omega_{p}^{2}%
}(t)\cdot(1-\tilde{t})^{2}dt & \text{otherwise.}%
\end{array}
\right.
\end{align*}

\begin{example}[Continued from Example \ref{ex:ex4}]\label{ex:ex5}
With the above formulae and $T$, $S$ as in Examples \ref{ex:ex3} and \ref{ex:ex4}, one may obtain the following:
\[
\begin{array}{lllll}
A_0\doteq0.258,&\quad&B_0\doteq0.229,&\quad&C_0\doteq0.193,\\
A_1\doteq0.518,&&B_1\doteq0.317,&&C_1\doteq0.159.
\end{array}
\]
Thus
\begin{align*}
M_0&\doteq0.406(2\sqrt{0.258\cdot 0.193}+0.229)\doteq0.274,\\
M_1&\doteq(1-0.406)(2\sqrt{0.518\cdot 0.159}+0.317)\doteq0.529.
\end{align*}
By Theorem \ref{thm:optimalZalpha}, we obtain
\begin{align*}
\alpha_0&\doteq1/(1+\sqrt{0.193/0.258})\doteq0.536,\\
\alpha_1&\doteq1/(1+\sqrt{0.159/0.518})\doteq0.643,
\end{align*}
and $z_1^*\doteq0.419$. Thus $\alpha\doteq(0.536,0.643)$ and $Z\doteq(0,.419,1)$.
One may further calculate
\[u_{q\circ m^{*}}\doteq\dfrac{1.249^2}{(\sqrt{0.274}+\sqrt{0.529})^2}\doteq0.997.\]
\end{example}


\section{Algorithm}

\label{sec:AlgorithmExample}

In this section, we summarize the above ideas and results as
Algorithm~\ref{alg:radreparam} and illustrate how the algorithm works for the cubic curve in Example \ref{ex:ex1}.

\begin{algorithm}[htb]
\caption{\sf Optimal\_Radical\_Transformation} \label{alg:radreparam} \rm\mbox{}
\begin{glists}{0em}{4em}{0em}{0.2em}
\iteml{Input:}  $p$, a rational parameterization of a plane curve.
\iteml{Output:} $r$, the optimal piecewise radical transformation of $p$ such that $u_{p\circ r}>u_p$.
\end{glists}
\begin{glists}{0em}{2em}{0em}{0.2em}
\iteml{1.} Compute $\omega_p$ and $\mu_p$ using \eqref{eq:w}, $u_p$ using \eqref{eq:up} and $\omega_p'$.
\iteml{2.} Solve $\omega_p\omega_p'=0$ and get $T$.
\iteml{3.} Compute $S$,
$Z$, $\alpha$ and $u$ using \eqref{eqs:si} and \eqref{eqs:alphaZ}.
\iteml{4.} Construct $\varphi$ with $T$, $S$ and $m$ with $S$, $Z$, $\alpha$ using \eqref{eq:varphii} and \eqref{eqs:moebiuspiece}.
\iteml{5.} $r \leftarrow  \varphi \circ m$.
\iteml{6.} Return $r$.
\end{glists}
\end{algorithm}

\begin{example}[Continued from Example \ref{ex:ex5}] Given $p=(t,t^3)$, after the above calculation, one may obtain
\[T\doteq(0,0.439,1),\ \
S\doteq(0,0.406,1),\ \
Z\doteq(0,.419,1),\ \
\alpha\doteq(0.536,0.643).\]
Then one may construct $\varphi$ with $T$ and $S$, and $m$ with $S$, $Z$ and $\alpha$, and obtain
\begin{align*}
\varphi\doteq&
\begin{cases}
0.688\sqrt{s}&\text{if}\quad 0.000\le s\le 0.406;\\
0.055+0.945s&\text{if}\quad 0.406\le s\le 1.000;
\end{cases}\\
m\doteq&\begin{cases}
\dfrac{-0.450z}{0.172z-0.536}&\text{if}\quad 0.000\le z\le 0.049;\\[8pt]
-\dfrac{0.165z+0.192z}{0.492z-0.849}&\text{if}\quad 0.419\le z\le 1.000.
\end{cases}
\end{align*}
Then
the optimal transformation $r$ is constructed below.
\[r=\varphi\circ m\doteq
\begin{cases}
0.462\,\sqrt {-{\dfrac {z}{ 0.172z- 0.536}}} & \text{if}\quad 0.000\le z\le 0.419;\\[10pt]
\dfrac{-0.129z - 0.228}{0.492z - 0.849} & \text{if}\quad 0.419\le z\leq  1.000.
\end{cases}\]
With the optimal radical transformation $r$, one may construct $p\circ r$ and obtain
\[p\circ r\doteq
\begin{cases}
\left({\dfrac {- 0.079\,\sqrt {z} \left( z- 3.116 \right) }{ \left( - 0.172\,z+ 0.5359 \right) ^{3/2}}}, {\dfrac {0.098\,{z}^{3/2}}{ \left( - 0.172\,z+ 0.536 \right) ^{3/2}}}\right) &\text{if}\quad  0.000\le z\le 0.419;\\[10pt]
\left({\dfrac {- 0.129\,z- 0.228}{ 0.492\,z- 0.849}},- {\dfrac { 0.002\,\left( z+ 1.771 \right) ^{3}}{ \left(  0.492\,z- 0.849 \right) ^{3}}}\right) & \text{if}\quad 0.419\le z\leq  1.000.
\end{cases}\]
The angular speed function of $p\circ r$ is
\[w_{p\circ r}\doteq\begin{cases}\dfrac{0.781}{{z}^{2}- 0.420\,z+ 0.655}& \text{if}\quad 0.000\le z\leq  0.4187;\\[10pt]{\dfrac { - 1.379\,\left( z+ 1.771 \right)  \left( z- 1.725 \right) }{ \left( {z}^{2}- 1.456\,z+ 0.899 \right)  \left( {z}^{2}- 4.876\,z+ 9.888 \right) }} & \text{if}\quad 0.4187\le z\leq  1.000.\end{cases}\]
Furthermore, one may calculate its uniformity $u_{p\circ r}\doteq0.997$.

The plots of $p$ and $p\circ r$ as well as the behavior of their angular speed functions are shown below:

\medskip

\noindent \includegraphics[width=2.0in]{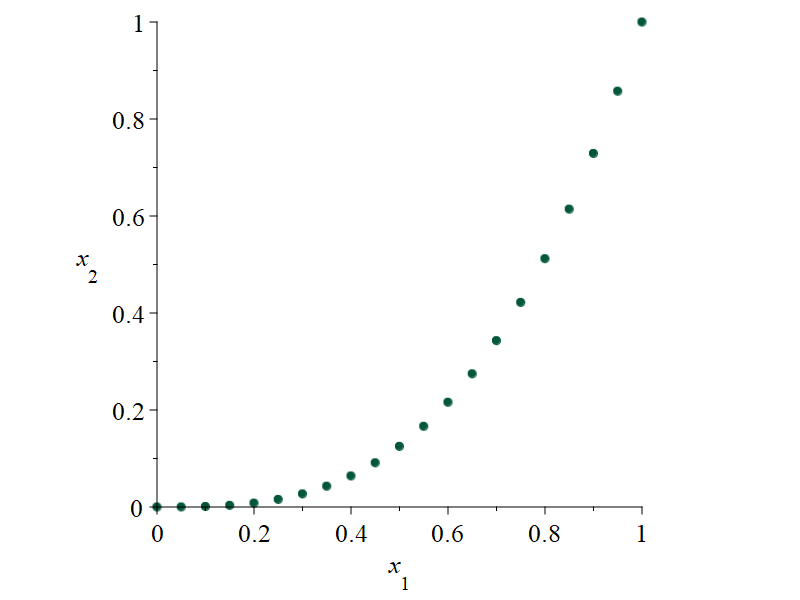} ~~
\includegraphics[width=2.0in]{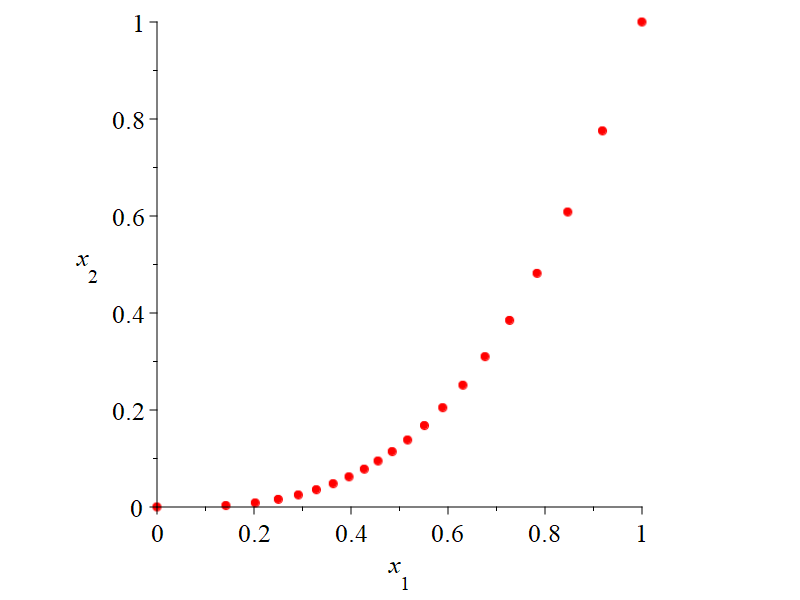} ~~
\includegraphics[width=2.0in]{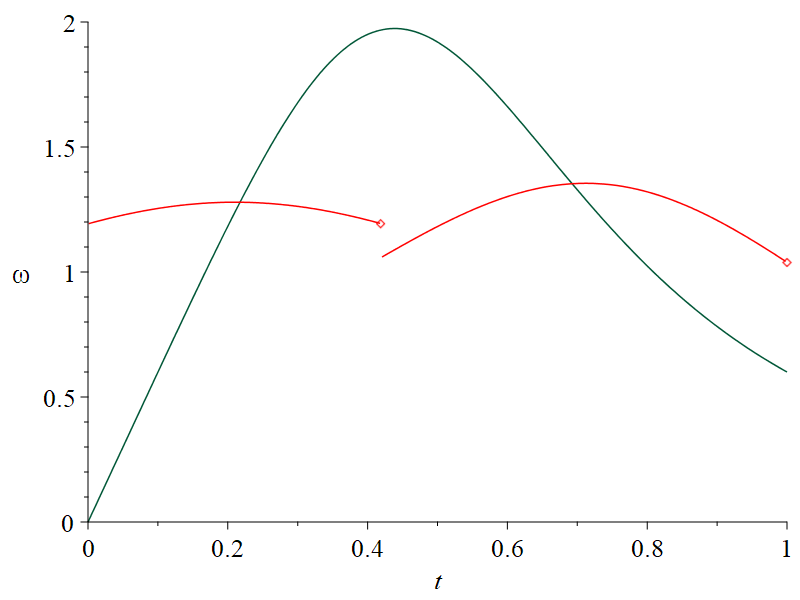}

\noindent where
\begin{itemize}
\item
the left plot shows the equi-sampling of the original parameterization $p$
(green);

\item
the middle plot shows the equi-sampling of the optimal piecewise radical reparameterization $p\circ r$ (red);

\item
the right plot shows the angular speed functions of $p$  and
$p\circ r$.
 \end{itemize}

\noindent It is seen that the angular speed uniformity is greatly improved by the piecewise
radical reparameterization.
\end{example}

\section{Implementational Issues/Suggestions}
\label{sec:implementation}

If one chooses to implement the proposed algorithm  using floating-point
arithmetic, then, as usual, one should be careful to avoid numerical instability.

For instance, if  $L_{i}$ is computed by
\eqref{eqs:mk} naively, then it leads to instability.
For example, $t_{i}=1/\sqrt{2}$ is a zero of $\omega_{p}$ with
multiplicity $1$, so
\[
\omega_{p}(t)=|t-1/\sqrt{2}|\cdot{m}(t),
\]
where ${m}(1/\sqrt{2})\neq0$. The numeric solution over $[0,1]$ is
$t_{i}=0.707$. Thus
\[
L_{i}\, =\,{\Delta t_{i}}\int_{t_{i}}^{t_{i+1}}\dfrac{\omega_{p}^{2}%
(t)}{(\mu_{i}+1)\tilde{t}^{\mu_{i}}}\,dt  =\,\Delta t_{i}^{2}\int_{t_{i}}^{t_{i+1}}\dfrac{(t-1/\sqrt{2})^{2}\cdot
{m}^{2}(t)}{2\,(t-0.707)}\,dt.
\]
During integration, it is necessary to evaluate the integral at $t=0.707$.
When~$t$ approaches $0.707$, the integral quickly increases to $+\infty$,
causing numerical instability.

To avoid such cases, one could adopt a technique from
symbolic computation to represent algebraic numbers.
 Suppose that $t=\gamma$ is a zero of $\omega_{p}(t)$ with multiplicity
$\mu_{i}$ and $t_{i}$ is its numerical approximation. By \eqref{eq:w},
$\omega_{p}^{2}(t)$ is a rational function. Let $G$ and $H$ be its numerator
and denominator. Then $t=\gamma$ is a zero of $G$ with multiplicity
$2\,\mu_{i}$. Carrying out the Euclidean division with $G$ as the dividend and
$(t-\gamma)^{2\,\mu_{i}}$ as the divisor, we obtain
\[
G(t)=(t-\gamma)^{2\,\mu_{i}}Q(t,\gamma)+R(t,\gamma).
\]
Since $t=\gamma$ is a zero of $G$ with multiplicity $2\,\mu_{i}$, it is also a
zero of $R(t,\gamma)$ with multiplicity at least $2\,\mu_{i}$. Given that $\deg(R,
t)<2\,\mu_{i}$, $R(t,\gamma)$ must be zero, which leads to the following
conclusion:
\[
L_{i}\,=  \,{\Delta t_{i}}\int_{t_{i}}^{t_{i+1}}\dfrac{\omega_{p}^{2}%
(t)}{(\mu_{i}+1)\,\tilde{t}^{\mu_{i}}}\,dt\doteq  \,\dfrac{\Delta t_{i}^{\mu_{i}+1}}{\mu_{i}+1}\cdot\int_{t_{i}%
}^{t_{i+1}}\dfrac{Q(t,t_{i})(t-t_{i})^{\mu_{i}}}{H(t)}\,dt.
\]

\paragraph{Acknowledgements}

Hoon Hong's work was supported by National Science Foundations of USA under Grant No. 1813340 and
Jing Yang's work was supported by National Natural Science Foundation of China
under Grant Nos. 11526060 and 12261010.

\nocite{*}
\bibliographystyle{eptcs}
\bibliography{generic}
\end{document}